\documentclass{cccg26}
\usepackage{graphicx,amssymb,amsmath}
\usepackage[T1]{fontenc}
\usepackage{booktabs}
\usepackage{subcaption}
\usepackage{tikz}
\usetikzlibrary{angles,quotes}
\usetikzlibrary{decorations.pathreplacing,calligraphy,arrows.meta}
\usepackage{mathtools}
\usepackage{microtype}
\usepackage{comment}
\usepackage{todonotes}
\usepackage[compact]{titlesec}
\titlespacing*{\section}      {0pt}{4pt plus 1pt minus 1pt}{2pt plus 1pt minus 1pt}
\titlespacing*{\subsection}   {0pt}{3pt plus 1pt minus 1pt}{1pt plus 1pt minus 1pt}
\titlespacing*{\subsubsection}{0pt}{2pt plus 1pt minus 1pt}{1pt plus 1pt minus 1pt}
\titlespacing*{\paragraph}    {0pt}{2pt plus 1pt minus 1pt}{0.5em}
\usepackage{etoolbox}
\AtEndEnvironment{lemma}{\vspace{-6pt}}
\AtEndEnvironment{theorem}{\vspace{-6pt}}
\AtEndEnvironment{proof}{\vspace{-6pt}}
\newtheorem{observation}{Observation}
\newtheorem{remark}{Remark}
\newtheorem{definition}{Definition}

\title{Guaranteed Escape for a Bouncing Robot in Pipe Chains} 

\author{Yeganeh Bahoo\thanks{\raggedright Toronto Metropolitan University, Toronto, ON M5B~2K3, Canada, \texttt{\{bahoo,ahmad.kamaludeen,somnath.kundu\}@torontomu.ca}}
	\and
	Ahmad Kamaludeen\footnotemark[1]
	\and
	Somnath Kundu\footnotemark[1]}
\index{Bahoo, Yeganeh}
\index{Kundu, Somnath}
\index{Kamaludeen, Ahmad}

\begin{document}
\thispagestyle{empty}
\maketitle

\begin{abstract}
We study the symmetric bouncing of a point robot within orthogonally-joined rectangles with equal width, which we refer to as pipes. We provide an exhaustive case analysis of every trajectory pattern inside a single rectangular pipe segment, identifying the conditions under which the robot exits. We then extend the analysis to L-shaped pipes and, more generally, to linear chains of $k$ orthogonally connected pipe segments. We prove exit guarantees for the special angle $\alpha = \pi/4$. Furthermore, these results extend to pipes with curved joints. 
\end{abstract}

\let\thefootnote\relax\footnotetext{We acknowledge the support of the Natural Sciences and Engineering Research Council of Canada (NSERC).}

\section{Introduction}
A bouncing robot is a mobile robot that moves in straight lines and reflects off obstacles like a billiard ball. We consider such a robot within an orthogonal pipe: an environment composed of axis-aligned rectangular corridors connected orthogonally, where their openings are placed end to end (at the corners of the each rectangle in the pipe), so that there exactly two openings in the entire composite polygon; see Figure~\ref{fig:chain}. The robot has no sensors to detect openings except by physically reaching them with an angle $\alpha > 0$. When it hits a wall, it bounces symmetrically; if it hits a corner, it stops. We assume there is an opening of length equal to the pipe's width on the boundary through which the robot can escape. The central problem is to decide how the robot should move so that it is guaranteed to find the exit.

This problem connects with billiards and dynamical systems~\cite{Boldrighini1978,Tokarsky1995,gutkin1986billiards}. In robotics, bouncing robots have been explored as minimalistic agents for coverage and exploration, relevant to sensor-denied environments such as pipes and ducts~\cite{BlindBouncingRobots}. Kundu et al.~\cite{Kundu2023RectRoom} presented an algorithm guaranteeing escape from a single rectangle under co-prime side lengths. This result is extended to arbitrary side lengths by Kundu et al. in~\cite{Kundu2024Escape}. Escaping orthogonal pipe systems, however, remained open.

We start with an exhaustive case analysis of every trajectory pattern inside the End Block of a rectangular pipe segment (Table~\ref{tab:robot_cases_final}), which we define as the unit square at the pipes closed end, which one edge of which is the exit, and we identify the exact conditions on the initial angle $\alpha$ and position $y_0$ under which the robot exits. This enumeration is the technical foundation for handling junctions between pipe segments. We apply it to characterize the trajectories in L-shaped pipes (two corridors joined orthogonally) and extend to \emph{linear chains} of $k$ orthogonally connected segments, proving exit guarantees under various conditions on the angle and length of the segments.

\section{Background and Related Work}

\begin{definition}[Symmetric Bounce]\label{def:symmetric_bounce}
When a robot collides with a wall, it undergoes a \emph{symmetric bounce}. The \emph{incident angle} is the angle between the incoming trajectory and the wall's normal; the \emph{reflected angle} is the angle between the outgoing trajectory and the wall's normal. A symmetric bounce requires the incident angle to equal the reflected angle, following the classical specular reflection law.
\end{definition}

\paragraph{Robot model.}
We model the robot as a \emph{point} moving at unit speed in a two-dimensional orthogonal pipe environment with no interior obstacles. The robot has no sensors to detect openings or measure distances; its only interaction with the environment is through collisions with the boundary. The robot traverses a straight-line trajectory until it collides with a wall or corner. Upon hitting a corner, the robot terminates its motion. When colliding with a wall, the robot undergoes a symmetric bounce (Definition~\ref{def:symmetric_bounce}). Following reflection, the robot resumes straight-line motion along the new trajectory determined by the reflection angle until it encounters another wall or corner.
This behaviour follows the classical billiard reflection law~\cite{gutkin1986billiards}.

\paragraph{The unfolding method.}
A key concept in analyzing bouncing trajectories is the \emph{unfolding method}~\cite{TabachnikovBilliards}: instead of reflecting the robot's path at each wall, one reflects the \emph{rectangle} across the wall, producing an infinite tiling of the plane by copies of the rectangle. In this unfolded plane, the robot's trajectory becomes a single straight line, and openings on the boundary correspond to periodically spaced ``holes'' in the tiling. See Figure~\ref{fig:unfold} and Figure~\ref{fig:tiling} in Appendix B for reference. The robot escapes if and only if its trajectory intersects one of the holes. This is closely related to classical equidistribution results for billiard trajectories~\cite{Masur1982Ergodic}.

\paragraph{Related work.}
The foundational paper~\cite{BlindBouncingRobots} formalized bouncing robots as minimalistic agents that interact with the environment only through collisions, establishing bouncing as a viable primitive for coverage and navigation. More recently,~\cite{WildBodies} studied how uncontrolled or partially controlled dynamical systems can be steered by exploiting boundary interactions; these strategies can be viewed through the lens of universal plans~\cite{TimLavLav}.
Building on the unfolding method, Kundu et al.~\cite{Kundu2023RectRoom} proved that for co-prime side lengths, a carefully chosen angle ensures consecutive bounce distances never exceed~$1$, so any unit-length gap is detected; \cite{Kundu2024Escape} subsequently removed the co-primality requirement. Beyond single rooms,~\cite{Kundu2022Rectilinear} showed that a bouncing robot with constant memory can systematically traverse all corridors of a rectilinear polygon.
Nilles et al.~\cite{Nilles2021Nondeterministic} relaxed deterministic reflection, analyzing \emph{nondeterministic} bouncing with angular uncertainty. Visibility with reflection in orthogonal polygons~\cite{Aronov1995VisibilityReflection,Biro1992Illumination} and specular-reflection visibility in general polygons~\cite{Vaezi2025Visibility} provide complementary geometric intuition for why repeated reflections can effectively search an environment.
The same specular-reflection dynamics underlie classical optics problems: O'Rourke and Petrovici~\cite{ORourke2001Narrowing} study how a beam, which moves symmetrically like the robot, can be progressively narrowed by a sequence of mirrors, and Eppstein~\cite{Eppstein2022Reflections} shows that the eventual fate of a light ray in an axis-parallel-and-diagonal ``mirror maze'' is decidable in polynomial time. These results are close in spirit to our setting, where the pipe walls act as mirrors and the question is whether a ray reaches a designated opening; the diagonal-mirror analysis of~\cite{Eppstein2022Reflections} is especially relevant to the $\alpha=\pi/4$ model.
Our analysis builds on~\cite{Kundu2024Escape} but differs in focus: rather than constructing a global escape strategy, we provide an exhaustive End Block case analysis (Table~\ref{tab:robot_cases_final}) that is the technical foundation for handling pipe junctions, which is not addressed in prior work.
Our analysis builds on~\cite{Kundu2024Escape} but differs in focus: rather than constructing a global escape strategy, we provide an exhaustive End Block case analysis (Table~\ref{tab:robot_cases_final}) that is the technical foundation for handling pipe junctions, which is not addressed in prior work. 

\section{Problem Definition and Single-Segment Analysis}
\label{sec:terminology}

The pipe system consists of pipes interconnected orthogonally at their endpoints, forming a cohesive network referred to simply as a ``pipe chain''. The two terminal pipes each feature an opening, thereby providing the entire system with exactly two openings, one at each end.

The robot enters the system via one of these openings. The process terminates if the robot exits through either opening. The robot's objective is to traverse the system by entering through one opening and exiting through the other.

We model the robot as a point undergoing symmetric bounces (Definition~\ref{def:symmetric_bounce}) in a two-dimensional rectangular pipe of dimensions $L \times 1$ (scaled so the pipe width is~$1$), with $L \geq 2$.

We consider a pipe with an opening on the left side and an opening on the bottom at the right end (Figure~\ref{fig:bounce}). Let $l_0$ denote the horizontal distance from the left opening to the left edge of the right opening. The robot starts at height $y_0 \in [0,1]$ with angle $\alpha \in (0, \pi/2)$ measured clockwise from the vertical. Between consecutive bounces on opposite walls, the robot travels a horizontal distance of $\tan\alpha$.

\begin{figure}[h]
\centering
\begin{tikzpicture}[scale=0.55]
    \coordinate (LB) at (1,0);
    \coordinate (OL) at (9,0);
    \coordinate (OT) at (9,2);
    \coordinate (RB) at (11,0);
    \coordinate (LT) at (1,2);
    \coordinate (RT) at (11,2);
    \coordinate (LD) at (1,2.5);
    \coordinate (OD) at (9,2.5);
    \coordinate (RD) at (11,2.5);
    \coordinate (RTD) at (11.5,2);
    \coordinate (RBD) at (11.5,0);
    \coordinate (B) at (0,0);
    \coordinate (S) at (1,1);
    \fill[gray!15] (LB) rectangle (RT);
        \draw [ultra thick] (LB)--(OL);
        \draw [dashed](OL)--(RB);
        \draw [ultra thick] (LT)--(RT);
        \draw [ultra thick] (RT)--(RB);
        \draw [dashed](LT)--(LB);
        \draw [dashed](OT)--(OL);
        \draw [dashed] (B)--(LB);
        \draw[|-|] (LD)--(OD);
        \draw[|-|] (OD)--(RD);
        \draw[|-|] (RTD)--(RBD);
        \node at (4.5, 2.75) {\small $l_0$};
        \node at (10, 2.75) {\small $1$};
        \node at (11.75, 1) {\small $1$};
        \node[draw, fill, circle, inner sep=1pt] at (B) {};
        \node[draw, fill, circle, inner sep=1pt] at (S) {};
        \draw [dashed][thick][-stealth][blue] (B)--(S);
        \draw [thick][-stealth][blue] (S)--(2,2);
        \draw [thick][-stealth][blue] (2,2)--(4,0);
        \draw [dashed][thick][-stealth][blue] (4,0)--(5,1);
        \draw [dashed][thick][-stealth][blue] (5.5,1)--(6.5,0);
        \draw [thick][-stealth][blue] (6.5,0)--(8.5,2);
        \draw [thick][-stealth][blue] (8.5,2)--(9.5,1);
        \draw[|-|] (0,-0.75)--(1,-0.75);
        \node at (0.5,-1.5) {\small $y_0 \tan \alpha$};
        \draw[|-|] (6.5,-0.75)--(9,-0.75);
        \node at (7.75,-0.5) {\small $\epsilon_B$};
        \draw[|-|] (8.5,2.9)--(9,2.9);
        \node at (8.75,3.25) {\small $\epsilon_T$};
        \pic [draw, "$\alpha$", angle eccentricity=1.5] {angle = B--S--LB};
\end{tikzpicture}
\caption{A robot (blue arrows) symmetrically bouncing in a rectangular pipe.}
\label{fig:bounce}
\end{figure}
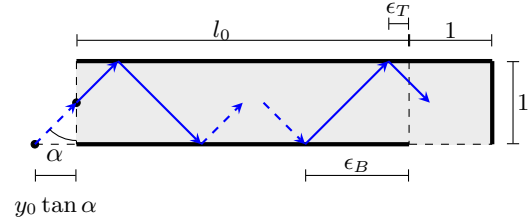

We define the \emph{End Block} (EB) as the unit square $[l_0,\, l_0+1]\times[0,1]$ at the right end of the pipe, whose bottom edge is the exit. We write $\operatorname{Rem}(a,b) = a - b\lfloor a/b\rfloor$. The horizontal distance from the last bounce on the bottom wall before the EB to the EB's left boundary is
$\epsilon_B = \operatorname{Rem}(y_0 \tan\alpha + l_0,\; 2\tan\alpha)$,
and the last bounce on the top wall is
$\epsilon_T = \operatorname{Rem}(\frac{(l_0 - (1-y_0)\tan\alpha)}{2\tan\alpha})$. See Figure \ref{fig:bounce}.
Since two adjacent sides of the EB are open, the robot can make at most two reflections within the EB before exiting or returning.

\begin{remark}\label{rmk:special_case}
The case $\alpha = \pi/4$ ($\tan\alpha = 1$) guarantees exit from any single rectangular pipe segment~\cite{Kundu2024Escape} and, consequently, from any pipe chain (Theorem~\ref{thm:pi4_chain}). It is therefore treated as a guaranteed-exit case.
\end{remark}

By Remark~\ref{rmk:special_case}, the case $\alpha = \pi/4$ is a guaranteed-exit case and excluded from the conditional classification in Table~\ref{tab:robot_cases_final}.

\subsection{End Block Classification}

The following theorem, whose proof appears in the Appendix, enumerates all possible trajectories inside the EB (see Figure~\ref{fig:cases}).

\begin{theorem}[Complete classification]\label{thm:classification}
Assume the robot has entered the End Block from either the bottom or the top.
Then exactly one of the nine bouncing patterns in Table~\ref{tab:robot_cases_final} occurs.
Moreover, the sequence Bottom $\to$ Side $\to$ Top $\to$ Exit (Case~4) is impossible.
\end{theorem}

\begin{table}[h]
\centering
\caption{End Block bounce classification.}
\label{tab:robot_cases_final}
\renewcommand{\arraystretch}{1.2}
\footnotesize
\begin{tabular}{@{}clcc@{}}
\toprule
\textbf{Case} & \textbf{Sequence} & \textbf{Exit?} & \textbf{Condition ($t = \tan\alpha$)}  \\
\midrule
1 & Bot$\to$Top & \checkmark & $t \le (1{+}\epsilon_B)/2$, \, $\alpha < \pi/4$ \\
2 & Bot$\to$Top$\to$Side & \checkmark & $(1{+}\epsilon_B)/2 < t \le 1{+}\epsilon_B$ \textsuperscript{$\dagger$} \\
3 & Bot$\to$Top$\to$Side & $\times$ & $(1{+}\epsilon_B)/2 < t \le 1{+}\epsilon_B$ \textsuperscript{$\dagger$} \\
4 & Bot$\to$Side$\to$Top & --- & \textbf{Impossible} \\
5 & Bot$\to$Side$\to$Top & $\times$ & $1{+}\epsilon_B < t \le 2{+}\epsilon_B$ \\
6 & Bot$\to$Side & $\times$ & $t > 2{+}\epsilon_B$ \\
7 & Top$\to$Side & $\times$ & $t \le (1{+}\epsilon_T)/2$ \\
8 & Top$\to$Side & \checkmark & $t > (1{+}\epsilon_T)/2$ \\
9 & Top$\to$(direct) & \checkmark & Depends on $x$-position and exit \\
\bottomrule
\end{tabular}
{\scriptsize $\dagger$\,Cases 2 \& 3 share the same $\alpha$-range; exit depends on starting point.}
\end{table}

The proof of Theorem~\ref{thm:classification} relies on two supporting lemmas. Lemma~\ref{lemma:first_collision} determines whether the first bounce is on the top wall ($\tan\alpha \le 1 + \epsilon_B$) or the side wall ($\tan\alpha > 1 + \epsilon_B$). Lemma~\ref{lemma:top_sequences} then determines, when the first bounce is on the top wall, whether the path exits immediately ($2\tan\alpha \le 1 + \epsilon_B$) or hits the side wall ($2\tan\alpha > 1 + \epsilon_B$). The impossibility of Case~4 is proved by showing the required conditions are contradictory. Full proofs appear in the Appendix.
\begin{lemma}[First Bounce]\label{lemma:first_collision}
The first bounce is on the top wall iff $\tan\alpha \le 1 + \epsilon_B$; on the right side wall iff $\tan\alpha > 1 + \epsilon_B$.
\end{lemma}
\vspace{-2mm}
\begin{lemma}[Top-First Sequences]\label{lemma:top_sequences}
If the first bounce is on the top wall, the path exits immediately (Case~1) iff $2\tan\alpha \le 1 + \epsilon_B$, and hits the side wall (Cases~2\,\&\,3) iff $2\tan\alpha > 1 + \epsilon_B$.
\end{lemma}

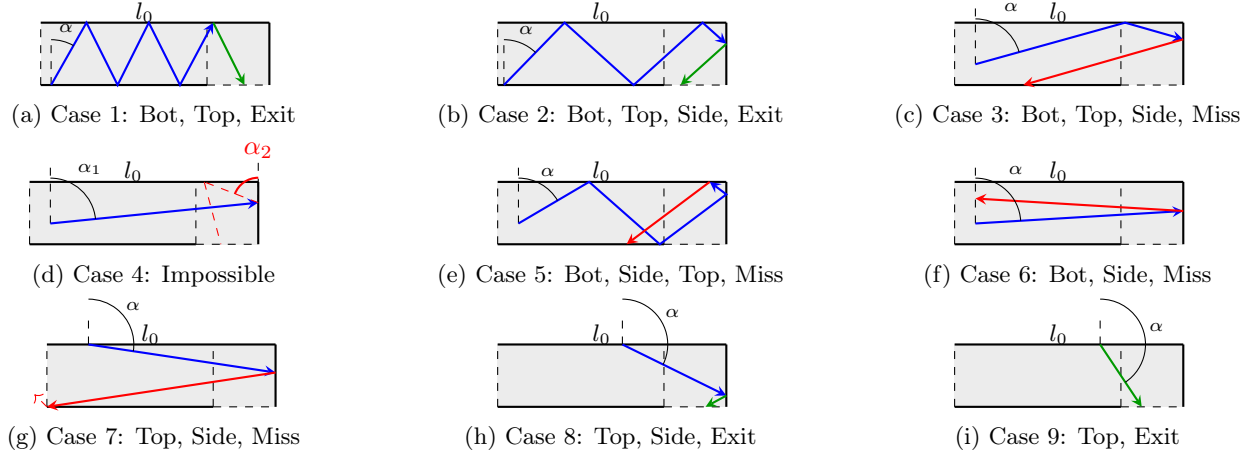
\begin{figure*}[t]
\centering
\tikzset{
    pipe/.style={thick, black},
    dashedline/.style={dashed, thin, black},
    traj/.style={thick, -stealth, blue},
    success/.style={thick, -stealth, green!60!black},
    fail/.style={thick, -stealth, red},
    ghost/.style={dashed, thin, red},
    anglelabel/.style={draw, "$\alpha$", angle radius=0.6cm, angle eccentricity=1.3, font=\scriptsize}
}
\def\DrawPipe{
    \def\L{4} \def\H{1.5} \def\End{5.5}
    \fill[gray!15] (0,0) rectangle (\End,\H);
    \draw[pipe] (0,\H) -- (\End,\H);
    \draw[pipe] (0,0) -- (\L,0);
    \draw[dashedline] (\L,0) -- (\End,0);
    \draw[dashedline] (0,0) -- (0,\H);
    \draw[pipe] (\End,0) -- (\End,\H);
    \draw[dashedline] (\L,0) -- (\L,\H);
    \node at (2.5, \H+0.25) {\small $l_0$};
   
}
\begin{subfigure}[t]{0.32\textwidth}
\centering
\begin{tikzpicture}[scale=0.55]
    \DrawPipe
    \coordinate (Start) at (0.25,0); \coordinate (VertUp) at (0.25,1);
    \coordinate (Top1) at (1.1,1.5); \coordinate (Bot1) at (1.85,0);
    \coordinate (Top2) at (2.6,1.5); \coordinate (Bot2) at (3.35,0);
    \coordinate (TopEB) at (4.15,1.5); \coordinate (Exit) at (4.9,0);
    \draw[dashedline] (Start)--(VertUp);
    \draw[traj] (Start)--(Top1)--(Bot1)--(Top2)--(Bot2)--(TopEB);
    \draw[success] (TopEB)--(Exit);
    \pic[anglelabel]{angle=Top1--Start--VertUp};
\end{tikzpicture}
\caption{Case 1: Bot, Top, Exit}
\end{subfigure}
\hfill
\begin{subfigure}[t]{0.32\textwidth}
\centering
\begin{tikzpicture}[scale=0.55]
    \DrawPipe
    \coordinate (Start) at (0.15,0); \coordinate (VertUp) at (0.15,1);
    \coordinate (Top1) at (1.61,1.5); \coordinate (Bot1) at (3.27,0);
    \coordinate (TopEB) at (4.93,1.5); \coordinate (Side) at (5.5,0.99);
    \coordinate (Exit) at (4.4,0);
    \draw[dashedline] (Start)--(VertUp);
    \draw[traj] (Start)--(Top1)--(Bot1)--(TopEB)--(Side);
    \draw[success] (Side)--(Exit);
    \pic[anglelabel]{angle=Top1--Start--VertUp};
\end{tikzpicture}
\caption{Case 2: Bot, Top, Side, Exit}
\end{subfigure}
\hfill
\begin{subfigure}[t]{0.32\textwidth}
\centering
\begin{tikzpicture}[scale=0.55]
    \DrawPipe
    \coordinate (Start) at (0.5,0.5); \coordinate (VertUp) at (0.5,2.0);
    \coordinate (TopEB) at (4.1,1.5); \coordinate (Side) at (5.5,1.1);
    \coordinate (Miss) at (1.65,0);
    \draw[dashedline] (Start)--(VertUp);
    \draw[traj] (Start)--(TopEB)--(Side);
    \draw[fail] (Side)--(Miss);
    \pic[anglelabel, pic text options={shift={(0pt,4pt)}}]{angle=TopEB--Start--VertUp};
\end{tikzpicture}
\caption{Case 3: Bot, Top, Side, Miss}
\end{subfigure}

\begin{subfigure}[t]{0.32\textwidth}
\centering
\begin{tikzpicture}[scale=0.55]
    \DrawPipe
    \coordinate (Start) at (0.5,0.5); \coordinate (VertUp) at (0.5,2.0);
    \draw[dashedline] (Start) -- (VertUp);
    \coordinate (Side) at (5.5,1.0);
    \draw[traj] (Start) -- (Side);
    \pic[draw, "$\alpha_1$", angle radius=0.6cm, angle eccentricity=1.3, font=\scriptsize, pic text options={shift={(0pt,4pt)}}]{angle = Side--Start--VertUp};
    \coordinate (VertUpSide) at (5.5,2.0);
    \draw[dashedline] (Side) -- (VertUpSide);
    \coordinate (TopHyp) at (4.2,1.5); \coordinate (ExitHyp) at (4.6,0.0);
    \draw[ghost] (Side) -- (TopHyp);
    \draw[ghost] (TopHyp) -- (ExitHyp);
    \draw[red, thick] (Side) ++(0,0.6) arc[start angle=90, end angle=160, radius=0.6];
    \path (VertUpSide) ++(20:0) node[red, yshift=4pt] {$\alpha_2$};
\end{tikzpicture}
\caption{Case 4: Impossible}
\end{subfigure}
\hfill
\begin{subfigure}[t]{0.32\textwidth}
\centering
\begin{tikzpicture}[scale=0.55]
    \DrawPipe
    \coordinate (Start) at (0.5,0.5); \coordinate (VertUp) at (0.5,2.0);
    \coordinate (Top1) at (2.2,1.5); \coordinate (Bot1) at (3.9,0);
    \coordinate (Side) at (5.5,1.2); \coordinate (Top2) at (5.1,1.5);
    \coordinate (Miss) at (3.1,0);
    \draw[dashedline] (Start)--(VertUp);
    \draw[traj] (Start)--(Top1)--(Bot1)--(Side)--(Top2);
    \draw[fail] (Top2)--(Miss);
    \pic[anglelabel]{angle=Top1--Start--VertUp};
\end{tikzpicture}
\caption{Case 5: Bot, Side, Top, Miss}
\end{subfigure}
\hfill
\begin{subfigure}[t]{0.32\textwidth}
\centering
\begin{tikzpicture}[scale=0.55]
    \DrawPipe
    \coordinate (Start) at (0.5,0.5); \coordinate (VertUp) at (0.5,2.0);
    \coordinate (Side) at (5.5,0.8); \coordinate (Miss) at (0.5,1.1);
    \draw[dashedline] (Start)--(VertUp);
    \draw[traj] (Start)--(Side);
    \draw[fail] (Side)--(Miss);
    \pic[anglelabel, pic text options={shift={(0pt,4pt)}}]{angle=Side--Start--VertUp};
\end{tikzpicture}
\caption{Case 6: Bot, Side, Miss}
\end{subfigure}

\begin{subfigure}[t]{0.32\textwidth}
\centering
\begin{tikzpicture}[scale=0.55]
    \DrawPipe
    \coordinate (StartTop) at (1.0,1.5); \coordinate (VertUp) at (1.0,2.1);
    \coordinate (Side) at (5.5,0.825); \coordinate (Miss) at (0.0,0.0);
    \draw[dashedline] (StartTop)--(VertUp);
    \draw[traj] (StartTop)--(Side);
    \draw[fail] (Side)--(Miss);
    \draw[dashed, red, ->] (Miss)--(-0.3,0.3);
    \pic[anglelabel]{angle=Side--StartTop--VertUp};
\end{tikzpicture}
\caption{Case 7: Top, Side, Miss}
\end{subfigure}
\hfill
\begin{subfigure}[t]{0.32\textwidth}
\centering
\begin{tikzpicture}[scale=0.55]
    \DrawPipe
    \coordinate (StartTop) at (3.0,1.5); \coordinate (VertUp) at (3.0,2.1);
    \coordinate (Side) at (5.5,0.27); \coordinate (Exit) at (5,0.0);
    \draw[dashedline] (StartTop)--(VertUp);
    \draw[traj] (StartTop)--(Side);
    \draw[success] (Side)--(Exit);
    \pic[anglelabel]{angle=Side--StartTop--VertUp};
\end{tikzpicture}
\caption{Case 8: Top, Side, Exit}
\end{subfigure}
\hfill
\begin{subfigure}[t]{0.32\textwidth}
\centering
\begin{tikzpicture}[scale=0.55]
    \DrawPipe
    \coordinate (StartTop) at (3.5,1.5); \coordinate (VertUp) at (3.5,2.1);
    \coordinate (Exit) at (4.5,0);
    \draw[dashedline] (StartTop)--(VertUp);
    \draw[success] (StartTop)--(Exit);
    \pic[anglelabel]{angle=Exit--StartTop--VertUp};
\end{tikzpicture}
\caption{Case 9: Top, Exit}
\end{subfigure}

\caption{All nine End Block bounce sequences. Green arrows indicate exits; red arrows indicate misses. Case~4 is proved impossible (Theorem~\ref{thm:classification}).}
\label{fig:cases}
\end{figure*}

\subsection{Periodicity and Density}
\label{sec:rect}

Following~\cite{Kundu2024Escape}, we recall two key facts about bouncing in rectangles that we will need for the chain analysis.
Let the pipe be the rectangle $[0,L]\times[0,H]$, and let a ray start at
$(0,y_0)$ with direction $\alpha \in (0,\pi/2)$.
Using these definitions we have Lemmas~\ref{lem:periodic}--\ref{lemma:shallow}
and Remark~\ref{rmk:aperiodic_dense}.

\begin{lemma}[Periodic Case~\cite{Kundu2024Escape}]\label{lem:periodic}
The sequence of $x$-coordinates at which the ray contacts the bottom boundary
is periodic iff $H\cot\alpha / L \in \mathbb{Q}$. Equivalently, there exist
coprime $p,q > 0$ with $\tan\alpha = (H/L)(q/p)$.
\end{lemma}

When the slope is rational with $s/(2L) = a/b$ in lowest terms
(where $s = 2H\cot\alpha$), the trajectory has period exactly~$b$; exit through
an opening of positive length is decidable by checking whether any of the $b$
bounce positions lies in the exit interval.

\begin{remark}[\cite{Kundu2024Escape}]\label{rmk:aperiodic_dense}
When $H\cot\alpha / L \notin \mathbb{Q}$, the bounce positions are aperiodic
and dense in $[0,L]$. Consequently, any exit of positive length is
eventually reached.
\end{remark}

\begin{lemma}[Shallow Angle Exit]\label{lemma:shallow}
If a bouncing robot travels inside a rectangular pipe segment with a relative angle $\alpha \le \pi/4$, it is guaranteed to reach the exit.
\end{lemma}
\begin{proof}
When $\alpha = \pi/4$, this follows from~\cite{Kundu2022Rectilinear}. When $\alpha < \pi/4$, we have $2\tan\alpha < 2 \le 1 + \epsilon_B + 1$, so the robot falls into Case~1 or Case~9 of Table~\ref{tab:robot_cases_final}, where it always exits regardless of~$y_0$.
\end{proof}

\section{Bouncing in Pipe Chains}

We now analyze linear chains of orthogonally connected pipe segments, which subsume L-shaped pipes as the case $k=2$.

\begin{definition}[Linear Pipe Chain]\label{def:chain}
A \emph{linear pipe chain} of length $k$ is $\mathcal{P} = R_1 \cup \cdots \cup R_k$ where each $R_i$ is a rectangle of width $1$ and length $L_i \ge 2$. Odd-indexed segments are horizontal and even-indexed segments are vertical. Adjacent segments share a unit-length \emph{junction} $J_i$, and the chain has terminal openings at the free ends of $R_1$ (entrance) and $R_k$ (exit).
\end{definition}

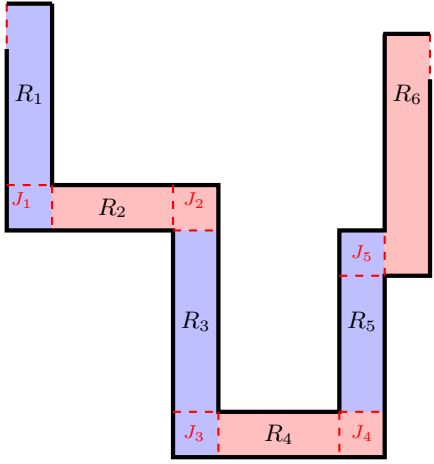
\begin{figure}[h]
\centering
\begin{tikzpicture}[scale=0.4, >=stealth]

    \fill[blue!25]  (0,-6.5) rectangle (1.5,1);        
    \fill[red!25]   (1.5,-5) rectangle (7,-6.5);     
    \fill[blue!25]  (5.5,-6.5) rectangle (7,-14); 
    \fill[red!25]   (7,-12.5) rectangle (12.5,-14); 
    \fill[blue!25]  (11,-12.5) rectangle (12.5,-6.5); 
    \fill[red!25]   (12.5,-8) rectangle (14,0);

    \draw[dashed, red, thick] (0,1) -- (0, -0.5);  
    \draw[ultra thick] (0,1) -- (1.5,1);
    \draw[dashed, red, thick] (14,0) -- (14,-1.5);  
    \draw[ultra thick] (12.45, 0) -- (14, 0);

    \draw[ultra thick] (0,-0.5) -- (0,-6.5) -- (5.5,-6.5) -- (5.5,-14) -- (12.5,-14) -- (12.5,-8) -- (14,-8) -- (14,-1.5);

    \draw[ultra thick] (12.5,0) -- (12.5,-6.5) -- (11,-6.5) -- (11,-12.5) -- (7,-12.5) -- (7,-5) -- (1.5,-5) -- (1.5,1);

    \draw[dashed, red, thick] (0,-5) -- (1.5,-5);
    \draw[dashed, red, thick] (1.5,-5) -- (1.5,-6.5);

    \draw[dashed, red, thick] (5.5,-5) -- (5.5,-6.5);
    \draw[dashed, red, thick] (5.5,-6.5) -- (7,-6.5);

    \draw[dashed, red, thick] (5.5,-12.5) -- (7,-12.5);
    \draw[dashed, red, thick] (7,-12.5) -- (7,-14);

    \draw[dashed, red, thick] (11,-12.5) -- (11,-14);
    \draw[dashed, red, thick] (11,-12.5) -- (12.5,-12.5);

    \draw[dashed, red, thick] (11,-8) -- (12.5,-8);
    \draw[dashed, red, thick] (12.5,-8) -- (12.5,-6.5);

    \node[red] at (0.5, -5.5) {\scriptsize $J_1$};
    \node[red] at (6.2, -5.5) {\scriptsize $J_2$};
    \node[red] at (6.2,-13.2) {\scriptsize $J_3$};
    \node[red] at (11.75,-13.2) {\scriptsize $J_4$};
    \node[red] at (11.75,-7.25) {\scriptsize $J_5$};

    \node at (0.75, -2)    {\small $R_1$};
    \node at (3.5, -5.75)  {\small $R_2$};
    \node at (6.25, -9.5)  {\small $R_3$};
    \node at (9, -13.25)   {\small $R_4$};
    \node at (11.75, -9.5) {\small $R_5$};
    \node at (13.25, -2)   {\small $R_6$};

\end{tikzpicture}
\caption{A linear pipe chain with $k=6$ segments. Dashed red segments indicate openings and junction boundaries.}
\label{fig:chain}
\end{figure}

The simplest case $k=2$ is the L-shaped pipe (Figure~\ref{fig:pipe}): a horizontal segment $R_1$ joined orthogonally to a vertical segment $R_2$. When the robot crosses the junction, the wall orientations swap and the relative angle is complemented (formalised in Observation~\ref{obs:angle} below). Thus if the robot exits $R_1$ with angle $\alpha > \pi/4$, the relative angle in $R_2$ becomes $\pi/2 - \alpha < \pi/4$, and by Lemma~\ref{lemma:shallow} exit from $R_2$ is guaranteed.

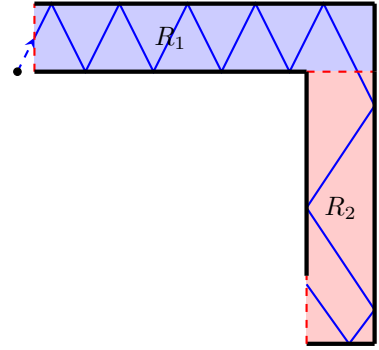
\begin{figure}[h]
\centering
\begin{tikzpicture}[scale=0.45]

    \fill[blue!20] (1,0) rectangle (11,2);       
    \fill[red!20]  (9,-8) rectangle (11,0);      

    \coordinate (S) at (1,1);
    \draw[dashed, thick, blue, -stealth] (0.5,0)--(S);
    \node[draw, fill, circle, inner sep=1pt] at (0.5,0) {};
    \draw[thick, blue] (S) -- (1.5,2) -- (2.5,0) -- (3.5,2) -- (4.5,0) -- (5.5,2) -- (6.5,0) -- (7.5,2) -- (8.5,0) -- (9.5,2) -- (11,-1);
    \draw[thick, blue] (11,-1) -- (9,-4) -- (11,-7) -- (10.25,-8) -- (9,-6.25);

    \draw[dashed, red, thick]  (1,0)  -- (1,2);    
    \draw[ultra thick] (1,0)  -- (9,0);              
    \draw[dashed, red, thick]  (9,0)  -- (11,0);    
    \draw[ultra thick] (1,2)  -- (11,2);            
    \draw[ultra thick] (11,2) -- (11,0);

    \draw[ultra thick] (11,0)  -- (11,-8);          
    \draw[ultra thick] (11,-8) -- (9,-8);            
    \draw[dashed, red, thick]  (9,-8)  -- (9,-6);   
    \draw[ultra thick] (9,-6)  -- (9,0);

    \node at (5,1) {$R_1$};
    \node at (10,-4) {$R_2$};
\end{tikzpicture}
\caption{An L-shaped pipe ($k=2$) with the robot's path in blue. Dashed red segments indicate openings (entrance, junction, and exit).}
\label{fig:pipe}
\end{figure}

\subsection{Angle Propagation}

When the robot crosses a junction, the wall orientations swap, complementing the angle. We first record this single-junction property, then generalize.

\begin{observation}\label{obs:angle}
When a bouncing robot moves from $R_i$ to $R_{i+1}$, its \emph{relative reflection angle} changes by $\pi/2$. More precisely, if the robot travels in $R_i$ with relative angle $\alpha$ measured with respect to the horizontal walls of $R_i$, then upon entering $R_{i+1}$ the corresponding relative angle (measured with respect to the walls of $R_{i+1}$) equals
\[
\alpha' \;=\; \frac{\pi}{2}-\alpha.
\]
\end{observation}

\begin{lemma}[Angle Alternation]\label{lem:angle_chain}
Let the robot enter $R_1$ with relative angle $\alpha \in (0, \pi/2)$.  The relative angle in $R_i$ is $\alpha_i = \alpha$ if $i$ is odd and $\alpha_i = \pi/2 - \alpha$ if $i$ is even. In particular, if $\alpha = \pi/4$ then $\alpha_i = \pi/4$ for all~$i$.
\end{lemma}
\begin{proof}
At each junction $J_i$, the two wall directions are exchanged because $R_{i+1}$ is oriented orthogonally to $R_i$. By Observation~\ref{obs:angle}, the relative angle transforms as $\alpha_{i+1} = \frac{\pi}{2} - \alpha_i$. The closed-form expression follows by induction: if $\alpha_1 = \alpha$, then $\alpha_2 = \pi/2 - \alpha$, $\alpha_3 = \pi/2 - (\pi/2 - \alpha) = \alpha$, and so on.
\end{proof}

\subsection{Transition Maps}

The unfolding method yields an explicit formula for the robot's entry position into the next segment, generalising the single-segment analysis.

\begin{definition}[Segment Transition Map]\label{def:transition}
For segment $R_i$ with length $L_i$ and relative angle $\alpha_i$, define $T_i : [0,1] \to [0,1]$ by
$T_i(y_i) = 1 - | ((y_i + L_i \tan\alpha_i) \bmod 2) - 1 |$,
provided the robot exits $R_i$ through junction~$J_i$.
\end{definition}

\begin{theorem}[Entry Point Propagation]\label{thm:propagation}
If the robot exits $R_i$ through $J_i$, the entry position into $R_{i+1}$ is $y_{i+1} = T_i(y_i)$.
\end{theorem}
\begin{proof}
The proof applies the unfolding method to eliminate reflections at the top and bottom walls of $R_i$. In the unfolded domain, the robot travels along the straight line $y(x) = y_i + x \tan\alpha_i$. At the far end $x = L_i$, the unfolded transverse coordinate is $y^* = y_i + L_i \tan\alpha_i$. Folding back into $[0,1]$ via reduction modulo~$2$ and reflection gives $y_{i+1} = 1 - |(y^* \bmod 2) - 1|$, which is $T_i(y_i)$.
\end{proof}

\begin{definition}[Chain Transition Map]\label{def:chain-transition}
The \emph{chain transition map} from the entrance of $R_1$ to the entrance of $R_{j+1}$ is the composition
\[
  \Phi_j = T_j \circ T_{j-1} \circ \cdots \circ T_1.
\]
Thus $y_{j+1} = \Phi_j(y_0)$, provided the robot successfully exits each intermediate segment.
\end{definition}

\begin{remark}\label{rem:structure}
Each map $T_i$ is piecewise linear on $[0,1]$. The parameter $\delta_i = L_i \tan\alpha_i \bmod 2$ determines the translation amount. When $\delta_i$ is irrational, iteration of the underlying circle rotation produces an equidistributed sequence by Weyl's theorem. When $\delta_i$ is rational, the orbit is periodic.
\end{remark}

Figure~\ref{fig:tiling} illustrates the unfolding: a single traversal from entrance to exit corresponds to a straight line across one unfolded square of side length $L_1 + 2d + L_2$, and each additional pass corresponds to entering an adjacent copy, so after $n$ passes the trajectory is a line across $n$ stacked squares.

\subsection{Exit Guarantees}

\begin{theorem}[Exit for $\alpha = \pi/4$]\label{thm:pi4_chain}
A bouncing robot with initial angle $\alpha = \pi/4$ exits every segment of a linear pipe chain on its first pass.
\end{theorem}
\begin{proof}
By Lemma~\ref{lem:angle_chain}, the relative angle is $\pi/4$ in every segment. By Lemma~\ref{lemma:shallow}, a robot with angle $\alpha \leq \pi/4$ is guaranteed to exit any single rectangular pipe segment.

We proceed by induction on the number of segments~$k$.

\emph{Base case} ($k = 1$): The robot exits $R_1$ by the single-segment guarantee.

\emph{Inductive step}: Assume the robot exits any chain of $k-1$ segments. In a chain of $k$ segments, the robot exits $R_1$ and enters $R_2$ at some position $y_1 = T_1(y_0)$ with angle $\alpha_2 = \pi/4$. The remaining chain $R_2, \ldots, R_k$ is itself a linear chain of $k-1$ segments, and by the inductive hypothesis the robot exits it.
\end{proof}

For $\alpha \neq \pi/4$, not every segment is guaranteed on a single pass:

\begin{lemma}[Guaranteed vs.\ Conditional Segments]\label{lem:dichotomy}
Consider a linear pipe chain with initial angle $\alpha \in (0, \pi/2)$, $\alpha \neq \pi/4$.
Segments where the effective relative angle is at most $\pi/4$ are \emph{guaranteed segments}: the robot always exits on its first pass, regardless of entry position (by Lemma~\ref{lemma:shallow}).
Segments where the effective relative angle exceeds $\pi/4$ are \emph{conditional segments}: whether the robot exits depends on the entry position $y_i$ and the specific conditions of Table~\ref{tab:robot_cases_final}.

By Lemma~\ref{lem:angle_chain}: if $\alpha < \pi/4$, odd-indexed segments are guaranteed and even-indexed segments are conditional; if $\alpha > \pi/4$, even-indexed segments are guaranteed and odd-indexed segments are conditional.
\end{lemma}

\begin{proof}
Immediate from Lemma~\ref{lem:angle_chain} and Lemma~\ref{lemma:shallow}.
\end{proof}

\begin{theorem}[Exiting Irrational Slopes]\label{thm:irrational_chain}
Let the robot enter segment $R_i$ with relative angle $\alpha_i$ such that $\frac{\cot\alpha_i}{L_i} \notin \mathbb{Q}$. Then the sequence of bounce points on the boundary of $R_i$ is dense, and the robot's trajectory eventually intersects any exit of positive length.
\end{theorem}

\begin{proof}
By Remark~\ref{rmk:aperiodic_dense}, when $\frac{\cot\alpha_i}{L_i} \notin \mathbb{Q}$, the sequence of bounce positions is dense. Since any exit of positive length is a non-degenerate interval, the trajectory must eventually intersect it.
\end{proof}

\begin{cor}[Irrational Chain Exit]\label{cor:irrational_chain}
Suppose that for every segment $R_i$ in the chain, $\frac{\cot\alpha_i}{L_i} \notin \mathbb{Q}$. A sufficient condition is $\tan\alpha \notin \mathbb{Q}$ with all $L_i$ rational. Then the robot is guaranteed to exit the entire chain, regardless of initial position~$y_0$.
\end{cor}

\begin{proof}
By Theorem~\ref{thm:irrational_chain}, whenever the robot is in a segment with irrational effective slope, it will eventually exit through any opening of positive length. At each segment $R_i$, the robot may exit forward through $J_i$ or backward through $J_{i-1}$; by density the robot will eventually hit both exits. Since the chain is finite and acyclic, the robot cannot remain trapped in any bounded subsection indefinitely: the trajectory must eventually propagate to $R_k$ and exit.
\end{proof}

\begin{cor}[Rational Chain Decidability]\label{cor:rational_chain}
For a linear chain where every segment has rational effective slope $\frac{\cot\alpha_i}{L_i} \in \mathbb{Q}$, the question of whether the robot exits the chain is decidable. In each segment, the trajectory is periodic, so the set of possible entry points into each segment is finite. Since the chain is finite and each segment contributes finitely many states, reachability of the terminal exit is decidable by exhaustive enumeration.
\end{cor}

\begin{definition}[Single-Pass Traversal]\label{def:single_pass}
A robot \emph{traverses the chain in a single pass} if, starting from the entrance of $R_1$, it sequentially exits $R_1$ through $J_1$, then $R_2$ through $J_2$, and so on, reaching the terminal exit at $R_k$ without ever re-entering a previously visited segment.
\end{definition}

\begin{theorem}[Single-Pass Sufficient Conditions]\label{thm:single_pass}
The robot traverses the chain in a single pass if, for every segment $R_i$, the entry position $y_i$ and relative angle $\alpha_i$ satisfy an exit condition from Table~\ref{tab:robot_cases_final}. Concretely, single-pass traversal holds if:
\begin{enumerate}
  \item $\alpha_i \leq \pi/4$ for all $i$ (i.e., $\alpha = \pi/4$), or
  \item for each conditional segment $R_i$ (where $\alpha_i > \pi/4$), one of Cases~1, 2, 8, or~9 of Table~\ref{tab:robot_cases_final} holds, i.e., the End Block analysis yields an exit rather than a miss.
\end{enumerate}
The entry position $y_i$ is computed iteratively as $y_i = T_{i-1}(y_{i-1})$ via Theorem~\ref{thm:propagation}.
\end{theorem}

\begin{proof}
By induction on~$k$.
The base case $k=1$ is a single-segment exit: if $\alpha_1 \leq \pi/4$, this is guaranteed by Lemma~\ref{lemma:shallow}; if $\alpha_1 > \pi/4$, one of the exit cases from Table~\ref{tab:robot_cases_final} must hold by hypothesis.
For the inductive step, the robot exits $R_1$ and enters $R_2$ at $y_1 = T_1(y_0)$ with angle $\alpha_2 = \pi/2 - \alpha$. The remaining chain $R_2, \ldots, R_k$ is a chain of $k-1$ segments; by hypothesis the exit conditions hold for each, and by the inductive hypothesis the robot traverses them in a single pass.
\end{proof}

\subsection{Backtracking Dynamics}

When the robot fails to exit a conditional segment on its first pass, it returns to the previous segment through the junction it entered from. This gives rise to a richer dynamical system in which the robot may oscillate between segments before eventually exiting.

\begin{definition}[Return Map]\label{def:return}
If the robot enters $R_i$ at transverse position $y_i$ but misses the End Block exit, it returns through $J_{i-1}$ at some position $y_i^{\mathrm{ret}}$. Define the \emph{return map} $\overline{T}_i : [0,1] \to [0,1]$ by $\overline{T}_i(y_i) = y_i^{\mathrm{ret}}$.
\end{definition}

The return map $\overline{T}_i$ is computed by the same unfolding method used for the forward transition $T_i$. The total horizontal distance traveled in $R_i$ before return depends on which miss case from Table~\ref{tab:robot_cases_final} applies (Cases~3, 5, 6, or~7), each producing a different effective path length before the robot exits backward.

\begin{definition}[Global State]\label{def:global_state}
The \emph{global state} of the robot at any junction $J_i$ is the triple $(i, y, d)$, where $i$ is the junction index, $y \in [0,1]$ is the transverse position, and $d \in \{+, -\}$ indicates the direction of travel ($+$ for forward, $-$ for backward in the chain).
\end{definition}

\begin{theorem}[Composite Dynamical System]\label{thm:composite}
The robot's motion through the chain defines a dynamical system on the state space $\mathcal{S} = \{1, \ldots, k-1\} \times [0,1] \times \{+,-\}$. The dynamics are:
\begin{enumerate}
  \item If the robot is at $J_i$ traveling forward ($d=+$), it enters $R_{i+1}$. If it exits $R_{i+1}$ forward through $J_{i+1}$, the new state is $(i{+}1, T_{i+1}(y), +)$; if it misses and returns through $J_i$, the new state is $(i, \overline{T}_{i+1}(y), -)$.
  \item If the robot is at $J_i$ traveling backward ($d=-$), it re-enters $R_i$. If it exits $R_i$ backward through $J_{i-1}$, the new state is $(i{-}1, \overline{T}_i(y), -)$; if it exits forward through $J_i$, the new state is $(i, T_i(y), +)$.
\end{enumerate}
\end{theorem}

\begin{proof}
This follows from the definitions of $T_i$ and $\overline{T}_i$ combined with the End Block classification (Table~\ref{tab:robot_cases_final}), which determines whether the robot exits forward or backward from each segment.
\end{proof}

\begin{theorem}[Exit for Dense Trajectories]\label{thm:eventual}
Let $\alpha$ be such that $\tan\alpha \notin \mathbb{Q}$ and all segment lengths $L_i$ are rational. Then the robot eventually exits the chain through the terminal exit at~$R_k$.
\end{theorem}

\begin{proof}
Under these conditions, the effective slope is irrational in every segment. By Theorem~\ref{thm:irrational_chain}, the trajectory is dense in every segment the robot visits.

By the End Block analysis (Table~\ref{tab:robot_cases_final}), for any conditional segment $R_i$, the set of entry positions $y_i$ that yield a forward exit is a non-degenerate interval (or union of intervals) in $[0,1]$. We claim that the sequence of entry positions into any given segment over repeated passes is equidistributed. Each time the robot enters $R_i$, it enters at a position determined by the composed transition and return maps from previous segments. Since these maps are irrational rotations (modulo folding), the sequence of entry positions is dense in $[0,1]$ by Weyl's theorem. Therefore, there exists a pass on which $y_i$ falls in the exit interval, and the robot exits forward.

Applying this argument inductively along the chain: the robot eventually passes through $J_1$, then eventually through $J_2$, and so on, until it exits through the terminal exit of~$R_k$.
\end{proof}

The results of this subsection are summarized in Table~\ref{tab:chain_summary}.

\begin{table}[h]
\centering
\caption{Exit guarantees for a chain of $k$ pipe segments.}
\label{tab:chain_summary}
\renewcommand{\arraystretch}{1.2}
\footnotesize
\resizebox{\columnwidth}{!}{%
\begin{tabular}{@{}lp{3.2cm}l@{}}
\toprule
\textbf{Condition on $\alpha$} & \textbf{Guarantee} & \textbf{Ref.} \\
\midrule
$\alpha = \pi/4$ & Always exits (single pass) & Thm~\ref{thm:pi4_chain} \\
$\alpha < \pi/4$ & Odd guaranteed; even conditional & Lem~\ref{lem:dichotomy} \\
$\alpha > \pi/4$ & Even guaranteed; odd conditional & Lem~\ref{lem:dichotomy} \\
Irrational slope, $L_i \in \mathbb{Q}$ & Always exits (multi-pass) & Thm~\ref{thm:irrational_chain} \\
Rational slope & Decidable in finite time & Cor~\ref{cor:rational_chain} \\
\bottomrule
\end{tabular}}
\end{table}
\subsection{Bouncing within a Circular Pipe}

\begin{figure}[h]
    \centering
    \begin{tikzpicture}

        \def\radius{3}
        \def\thickness{0.75}

        \pgfmathsetmacro{\rinner}{\radius - \thickness}
        \pgfmathsetmacro{\router}{\radius + \thickness}

        \draw[thick, blue]
            (\radius + \thickness, 0)
            arc[start angle=0, end angle=90, radius=\radius + \thickness];

        \draw[thick, blue]
            (\radius - \thickness, 0)
            arc[start angle=0, end angle=90, radius=\radius - \thickness];

        \draw[thick, blue]
            (\radius - \thickness, 0) -- (\radius + \thickness, 0);
        \draw[thick, blue]
            (0, \radius - \thickness) -- (0, \radius + \thickness);

        \draw[red, dashed]
            (\radius, 0)
            arc[start angle=0, end angle=90, radius=\radius];

        \pgfmathsetmacro{\bounceAngle}{45}
        \pgfmathsetmacro{\bx}{\rinner * cos(\bounceAngle)}
        \pgfmathsetmacro{\by}{\rinner * sin(\bounceAngle)}

        \fill[green!70!black] (\bx, \by) circle (2pt);

        \pgfmathsetmacro{\tlength}{1.0}
        \draw[thick, green!70!black]
            ({\bx - \tlength*(-sin(\bounceAngle))},
             {\by - \tlength*cos(\bounceAngle)})
            --
            ({\bx + \tlength*(-sin(\bounceAngle))},
             {\by + \tlength*cos(\bounceAngle)});

        \pgfmathsetmacro{\nlength}{2.0}
        \pgfmathsetmacro{\nx}{cos(\bounceAngle)}
        \pgfmathsetmacro{\ny}{sin(\bounceAngle)}

        \draw[thick, black, dashed]
            ({\bx - \nlength*\nx}, {\by - \nlength*\ny})
            --
            ({\bx + \nlength*\nx}, {\by + \nlength*\ny});

        \pgfmathsetmacro{\startAngle}{35}
        \pgfmathsetmacro{\startX}{\router * cos(\startAngle)}
        \pgfmathsetmacro{\startY}{\router * sin(\startAngle)}

        \pgfmathsetmacro{\incdx}{\bx - \startX}
        \pgfmathsetmacro{\incdy}{\by - \startY}
        \pgfmathsetmacro{\incLen}{sqrt(\incdx*\incdx + \incdy*\incdy)}
        \pgfmathsetmacro{\incdxN}{\incdx / \incLen}
        \pgfmathsetmacro{\incdyN}{\incdy / \incLen}

        \draw[thick, orange, -latex]
            (\startX, \startY) -- (\bx, \by);
        \fill[orange] (\startX, \startY) circle (2pt);

        \pgfmathsetmacro{\dotprod}{\incdxN*\nx + \incdyN*\ny}
        \pgfmathsetmacro{\refdx}{\incdxN - 2*\dotprod*\nx}
        \pgfmathsetmacro{\refdy}{\incdyN - 2*\dotprod*\ny}

        \pgfmathsetmacro{\Bref}{2*(\bx*\refdx + \by*\refdy)}
        \pgfmathsetmacro{\Cref}{\bx*\bx + \by*\by - \router*\router}
        \pgfmathsetmacro{\Discref}{\Bref*\Bref - 4*\Cref}
        \pgfmathsetmacro{\tRef}{(-\Bref + sqrt(\Discref)) / 2}

        \pgfmathsetmacro{\refEndX}{\bx + \tRef*\refdx}
        \pgfmathsetmacro{\refEndY}{\by + \tRef*\refdy}

        \draw[thick, orange!80!red, -latex]
            (\bx, \by) -- (\refEndX, \refEndY);
        \fill[orange!80!red] (\refEndX, \refEndY) circle (2pt);

        \pgfmathsetmacro{\arcR}{0.30}

        \pgfmathsetmacro{\incRevAngle}{atan2(-\incdyN,-\incdxN)}

        \pgfmathsetmacro{\refAngleDeg}{atan2(\refdy,\refdx)}

        \draw[black]
            (\bx,\by) ++(\bounceAngle:\arcR)
            arc[start angle=\bounceAngle,
                end angle=\incRevAngle,
                radius=\arcR];
        \node[scale=0.75] at
            ({\bx + 0.25 + 0.45*cos((\bounceAngle+\incRevAngle)/2)},
             {\by + 0.15 + 0.45*sin((\bounceAngle+\incRevAngle)/2)})
            {$\theta$};

        \draw[black]
            (\bx,\by) ++(\bounceAngle:\arcR)
            arc[start angle=\bounceAngle,
                end angle=\refAngleDeg,
                radius=\arcR];
        \node[scale=0.75] at
            ({\bx + 0.15 + 0.45*cos((\bounceAngle+\refAngleDeg)/2)},
             {\by + 0.25 + 0.45*sin((\bounceAngle+\refAngleDeg)/2)})
            {$\theta$};

    \end{tikzpicture}
    \caption{Bouncing within a Circular Pipe}
    \label{fig:BouncingInCircularPipe}
\end{figure}
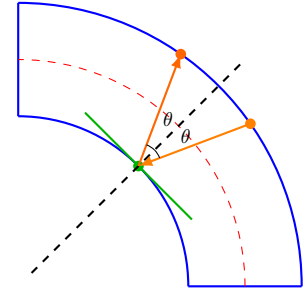

\begin{theorem}
If the pipes are connected with circular joints, a robot entering the pipe system will definitely exit.
\end{theorem}

\begin{proof}
Consider Figure~\ref{fig:BouncingInCircularPipe}. If a robot hits the inner wall of the pipe at an angle $\theta$ to the normal, it subtends an angle of $\frac{\pi}{2}-\theta$ with the tangent, and the reflected ray makes an angle of $\frac{\pi}{2}+\theta$ with the tangent. Note that it is impossible for the robot to approach exactly perpendicular to the tangent (i.e.\ $\theta = \frac{\pi}{2}$), since the joints are connected orthogonally to the pipe, nor can it travel exactly along the tangent (i.e.\ $\theta = 0$). Hence,
\[
    \frac{\pi}{2} > \theta > 0
    \implies
    \frac{\pi}{2} > \frac{\pi}{2} - \theta > 0,
\]
and furthermore,
\[
    \frac{\pi}{2} + \theta > \frac{\pi}{2} > \frac{\pi}{2} - \theta.
\]
Therefore, after every reflection off the inner wall, the reflected ray makes a strictly larger angle with the tangent than the incident ray. In particular, the robot can never return to the same side of the normal from which it bounced, so it must advance in the same direction after each inner-wall reflection.

By an analogous argument, if the robot instead bounces off the outer wall, it likewise continues to advance in the same direction. Consequently, within every circular joint, the robot strictly advances in its original direction of travel. Since no circular joint can reverse or redirect the robot, it continues to advance in the direction it entered. Thus, a robot traversing a pipe system connected by circular joints will always exit.
\end{proof}
\section{Conclusion}

We presented a complete classification of bouncing trajectories inside the End Block of a rectangular pipe segment (Table~\ref{tab:robot_cases_final}), and used it to analyze linear chains of $k$ orthogonally connected pipe segments. We proved that the bouncing angle alternates with period~2 along the chain (Lemma~\ref{lem:angle_chain}), that a robot at angle $\alpha = \pi/4$ always exits in a single pass (Theorem~\ref{thm:pi4_chain}), that irrational slopes guarantee exiting (Theorem~\ref{thm:irrational_chain}), that a robot always exits a pipe with a circular joint, and that rational slopes yield decidable exit (Corollary~\ref{cor:rational_chain}). Other works exploring this topic include \cite{Kundu2022Rectilinear,gutkin1986billiards,Kundu2023RectRoom}. Future directions include branching pipe systems (T-shaped pipes), cyclic networks, pipes of non-uniform width, and strategies for choosing optimal initial angles.


\small
\bibliographystyle{abbrv}
\bibliography{ref}

\newpage
\section*{Appendix}

\subsection*{A. Proofs for Single-Segment Analysis}

\begin{proof}[Proof of Lemma~\ref{lemma:first_collision}]
The horizontal distance required to reach height $y=1$ from $y=0$ is exactly $\tan\alpha$. Starting from $x=-\epsilon_B$, the horizontal coordinate at $y=1$ is $x_{top} = \tan\alpha - \epsilon_B$.
The top wall is bounded by $x \le 1$. Thus, the bounce occurs on the top wall $\iff$ $\tan\alpha - \epsilon_B \le 1 \iff \tan\alpha \le 1 + \epsilon_B$. Conversely, if $\tan\alpha - \epsilon_B > 1$, the path intersects $x=1$ at a height $y < 1$.
\end{proof}

\begin{proof}[Proof of Lemma~\ref{lemma:top_sequences}]
By symmetric bounce, the horizontal distance between two bottom contacts (or virtual bottom contacts) is $2\tan\alpha$. The projected landing point on the bottom axis is $x_{land} = 2\tan\alpha - \epsilon_B$.
\begin{itemize}
    \item The robot exits ($0 < x_{land} \le 1$) without hitting the side wall if $x_{land} \le 1$. This yields $2\tan\alpha - \epsilon_B \le 1 \iff 2\tan\alpha \le 1 + \epsilon_B$.
    \item If $x_{land} > 1$, the path intersects the Right Side wall ($x=1$) before reaching the bottom. In this case $2\tan\alpha > 1 + \epsilon_B$.
\end{itemize}
\end{proof}

\begin{proof}[Proof of Theorem~\ref{thm:classification} (Impossibility of Case~4)]
For a path to hit the Side wall first (Lemma~\ref{lemma:first_collision}), we require:
\begin{equation}
\label{eq:side_first}
\tan\alpha > 1 + \epsilon_B
\end{equation}
Upon hitting the side wall at height $h_r < 1$, the robot reflects leftward. The horizontal distance to the top wall is $\tan\alpha(1-h_r)$, and then to the bottom is $\tan\alpha$. The projected landing coordinate $x_{final}$ on the bottom axis (moving left from $x=1$) is given by:
\[
x_{final} = 1 - (2\tan\alpha - (1 + \epsilon_B)) = 2 + \epsilon_B - 2\tan\alpha
\]
For a valid Exit, we strictly require $x_{final} > 0$ (inside the End Block).
\[
2 + \epsilon_B - 2\tan\alpha > 0 \iff \tan\alpha < 1 + \frac{\epsilon_B}{2}
\]
However, Equation~(\ref{eq:side_first}) states $\tan\alpha > 1 + \epsilon_B$. Since $\epsilon_B \ge 0$, it is impossible to satisfy $\tan\alpha > 1 + \epsilon_B$ and $\tan\alpha < 1 + \epsilon_B/2$ simultaneously. Thus, Case~4 is impossible.
\end{proof}

\begin{proof}[Proof of Theorem~\ref{thm:classification} (Full)]
We determine the direction of travel from the reflection law: bouncing off a horizontal wall reverses the vertical component, while bouncing off a vertical wall reverses the horizontal component. The robot enters the End Block either from the bottom or the top.

\emph{Bottom entry.}
Suppose the robot enters $\mathcal{R}$ from the bottom (i.e., $\epsilon_B \le \tan\alpha$).
\begin{enumerate}
    \item \emph{Bottom $\to$ Top $\to$ Exit.}
    The robot bounces on the top wall and then immediately bounces into the bottom exit region iff
    \[
        0 < \tan\alpha \le \frac{1 + \epsilon_B}{2}.
    \]

    \item \emph{Bottom $\to$ Top $\to$ Side.}
    The robot bounces on the top wall and then bounces on the right side wall iff
    \[
        \frac{1 + \epsilon_B}{2} < \tan\alpha \le 1 + \epsilon_B.
    \]
    After the side bounce, the subsequent bottom bounce lies either:
    \begin{itemize}
        \item in the exit region $(0,1]$ (Exit), or
        \item in $(-\infty,0]$ (Miss).
    \end{itemize}

    \item \emph{Bottom $\to$ Side first.}
    If the first bounce is on the side wall, i.e.,
    \[
        \tan\alpha > 1 + \epsilon_B,
    \]
    then:
    \begin{itemize}
        \item Bottom $\to$ Side $\to$ Top $\to$ Exit is \textbf{impossible};
        \item Bottom $\to$ Side $\to$ Top $\to$ Miss occurs iff
              \[
                1 + \epsilon_B < \tan\alpha \le 2 + \epsilon_B;
              \]
        \item Bottom $\to$ Side $\to$ Side $\to$ Miss occurs iff
              \[
                \tan\alpha > 2 + \epsilon_B.
              \]
    \end{itemize}
\end{enumerate}

\emph{Top entry.}
Suppose the robot enters EB from the top. Let $\epsilon_T$ be the signed horizontal coordinate of the next virtual bottom bounce after the initial top bounce, which is exactly $\operatorname{Rem}\frac{(l_0 - (1-y_0)\tan\alpha)}{2\tan\alpha}$. 
\begin{enumerate}
    \item \emph{Top $\to$ Side $\to$ Exit.}
    The robot first bounces on the side wall and afterwards bounces into the bottom exit region iff
    \[
        0 < \epsilon_T \le \frac{1 + \epsilon_T}{2}.
    \]

    \item \emph{Top $\to$ Side $\to$ Miss.}
    The robot bounces on the side wall and then makes a bottom bounce outside the exit region iff
    \[
        \epsilon_T > \frac{1 + \epsilon_T}{2}.
    \]

    \item \emph{Top $\to$ Exit (no side bounce).}
    The robot travels directly into the exit. We can solve explicitly for the value of $\alpha$ in this case, given our initial point.
    We begin with an initial position $P = (P_x, 1)$.
    \[
        \alpha_{min} = \tan^{-1}\!\left(\frac{-1}{L-E - P_x}\right), \qquad
        \alpha_{max} = \tan^{-1}\!\left(\frac{-1}{L - P_x}\right).
    \]
\end{enumerate}

Thus every possible bounce sequence inside EB corresponds to exactly one row of Table~\ref{tab:robot_cases_final}, and the sequence Bottom $\to$ Side $\to$ Top $\to$ Exit is provably impossible.
\end{proof}
\subsection*{B. Figures}\label{apx:B}
\tikzset{
  realwall/.style={very thick, black},
  copywall/.style={thin, gray!60},
  sweep/.style={fill=brown!22},
  ghost/.style={fill=brown!12},
  traj/.style={very thick, red},
  zig/.style={very thick, red},
  mirror/.style={dash pattern=on 1pt off 1.5pt, gray!70},
  exitstyle/.style={very thick, green!55!black},
  ubrace/.style={decorate, decoration={calligraphic brace, amplitude=5pt}, thick}
}

A specular bounce off a wall is geometrically identical to letting the trajectory
continue \emph{straight} while reflecting the pipe across that wall. Since the robot
moves at $\pi/4$, every unit of travel advances it equally along the pipe and across
it, so traversing a pipe of length $L$ also accumulates $L$ of cross-travel, spread
over $L/d$ bounces (Fig.~\ref{fig:unfold}). Unfolding those bounces turns the zig-zag
into a single diagonal, and the $L/d$ reflected copies tile an $L\times L$ square whose
main diagonal is the unfolded path. For an L-shaped chain we unfold each arm in turn
(Fig.~\ref{fig:tiling}): the horizontal arm reflects across its top/bottom walls into an
$L_1\times L_1$ square, the corner contributes a $d\times d$ square, and the vertical arm
reflects across its left/right walls into an $L_2\times L_2$ square. The three squares
meet corner-to-corner along the diagonal of an $(L_1+d+L_2)$-square, so escape reduces
to asking whether that straight diagonal reaches the unfolded exit.
 
\begin{figure}[h]
\centering
\begin{tikzpicture}[scale=0.95]
  \def\d{1}\def\L{3}
  \foreach \k in {1,...,\numexpr\L-1\relax}{ \fill[ghost] (0,\k) rectangle (\L,\k+1); }
  \fill[blue!8] (0,0) rectangle (\L,\d);
  \foreach \k in {1,...,\numexpr\L-1\relax}{ \draw[copywall] (0,\k) rectangle (\L,\k+1); }
  \foreach \k in {1,...,\numexpr\L-1\relax}{ \draw[mirror] (0,\k) -- (\L,\k); }
  \draw[realwall] (0,0) rectangle (\L,\d);
  \draw[zig] (0,0) -- (1,1) -- (2,0) -- (3,1);
  \foreach \p in {(0,0),(1,1),(2,0),(3,1)}{ \fill[red] \p circle (1.6pt); }
  \draw[traj, dash pattern=on 5pt off 3pt] (0,0) -- (\L,\L);
  \fill[red] (\L,\L) circle (1.6pt);
  \node[font=\small, anchor=north] at (\L/2,-0.12) {real pipe ($L\times d$)};
  \node[font=\small, gray!45!black, anchor=west, align=left] at (\L+0.15,\L-0.55) {straight\\continuation};
  \node[font=\small, red, anchor=east] at (-0.1,0.5) {bounce};
  \node[font=\footnotesize, gray!45!black, anchor=west] at (\L+0.12,1.0) {fold lines};
  \draw[-{Stealth}, gray, very thick] (1.55,1.12) to[bend left=20] (2.0,1.7);
\end{tikzpicture}
\caption{Unfolding a single bounce. A $\pi/4$ trajectory (solid) that bounces off the
top/bottom walls is replaced by reflecting the pipe across each wall (dashed grey fold
lines); the path then continues as one straight diagonal (dashed red) through the
reflected copies. Over a length $L$ the copies fill an $L\times L$ square.}
\label{fig:unfold}
\end{figure}
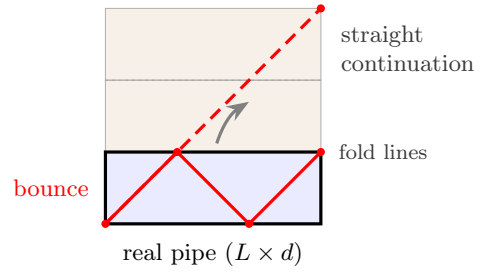
 
\begin{figure}[h]
\centering
\begin{tikzpicture}[scale=0.7]
  \def\d{1}\def\La{4}\def\Lb{3}
  \pgfmathsetmacro{\T}{\La+\d+\Lb}
  \fill[sweep] (0,0) rectangle (\La,\La);
  \fill[sweep] (\La,\La) rectangle (\La+\d,\La+\d);
  \fill[sweep] (\La+\d,\La+\d) rectangle (\T,\T);
  \foreach \k in {0,...,\numexpr\La-1\relax}{ \draw[copywall] (0,\k) rectangle (\La,\k+1); }
  \foreach \j in {0,...,\numexpr\Lb-1\relax}{ \draw[copywall] (\La+\d+\j,\La+\d) rectangle (\La+\d+\j+1,\T); }
  \draw[copywall] (\La,\La) rectangle (\La+\d,\La+\d);
  \fill[blue!8] (0,0) rectangle (\La+\d,\d);
  \fill[blue!8] (\La,-\Lb) rectangle (\La+\d,0);
  \draw[realwall] (0,\d) -- (\La+\d,\d) -- (\La+\d,-\Lb);
  \draw[exitstyle] (\La+\d,-\Lb) -- (\La,-\Lb);
  \draw[realwall] (\La,-\Lb) -- (\La,0) -- (0,0);
  \node[green!55!black, font=\small] at (\La+\d/2,-\Lb-0.45) {exit};
  \draw[exitstyle, dash pattern=on 4pt off 2pt] (\La+\d,\T) -- (\T,\T);
  \draw[traj] (0,0) -- (\T,\T);
  \fill[red] (0,0) circle (2.2pt);
  \node[red, font=\small, anchor=north east] at (-0.05,-0.05) {entry};
  \draw[ubrace] (\La,-0.15) -- (0,-0.15) node[midway,below=4pt,font=\small]{$L_1$};
  \draw[ubrace] (\La,\T+0.15) -- (\La+\d,\T+0.15) node[midway,above=4pt,font=\small]{$d$};
  \draw[ubrace] (\La+\d,\T+0.15) -- (\T,\T+0.15) node[midway,above=4pt,font=\small]{$L_2$};
  \draw[ubrace] (-0.2,0) -- (-0.2,\T) node[midway,left=4pt,font=\small]{$L_1+d+L_2$};
\end{tikzpicture}
\caption{Unfolding an L-shaped pipe (real pipe in bold/blue, with entry and \textcolor{green!55!black}{exit}).
The horizontal arm reflects upward into an $L_1\times L_1$ block (horizontal divisions), the corner into a
$d\times d$ square, and the vertical arm reflects rightward into an $L_2\times L_2$ block (vertical
divisions); each thin grey cell is one reflected copy of a pipe section. The straightened
trajectory (red) is the main diagonal of the resulting $(L_1+d+L_2)$-square; it escapes iff it
meets the unfolded image of the exit (dashed green).}
\label{fig:tiling}
\end{figure}
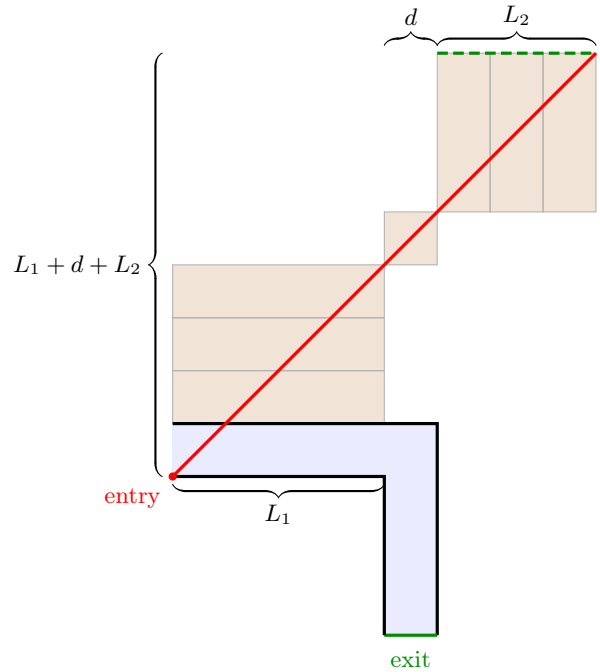

\end{document}